\documentclass[runningheads]{llncs}
\usepackage{graphicx}
\usepackage{ifthen}
\usepackage{xspace}
\usepackage{algorithm}
\usepackage{algorithmic}
\usepackage{xcolor}
\usepackage{amsmath, amssymb, amsfonts}
\usepackage{fullpage}

\begin{document}

%
%
%

\title{Multiagent Learning for Competitive Opinion Optimization}

\author{Po-An Chen\inst{1}\thanks{Corresponding author, supported in part by MOST 110-2410-H-A49-011-} \and Chi-Jen Lu\inst{2} \and Chuang-Chieh Lin\inst{3}\thanks{Co-corresponding author, supported in part by MOST 110-2222-E-032-002-MY2.} \and Ke-Wei Fu\inst{1}}
\authorrunning{P.-C. Chen, C.-J. Lu, and C.-C. Lin, and K.-W. Fu}
%
\institute{Institute of Information Management, National Yang Ming Chiao Tung University
, Taiwan \\\email{\{poanchen,andrewfu.mg09\}@nycu.edu.tw}\and
Institute of Information Science, Academia Sinica
, Taiwan 
\\\email{cjlu@iis.sinica.edu.tw}\and
Department of Computer Science, Tamkang University
, Taiwan 
\\\email{josephcclin@gms.tku.edu.tw}}
\maketitle
\begin{abstract}
From a perspective of designing or engineering for opinion formation games in social networks,
the \emph{opinion maximization (or minimization)} problem has been studied mainly for designing subset selecting algorithms.
We define a two-player zero-sum Stackelberg game of competitive opinion optimization by letting the player under study as the leader minimize the sum of expressed opinions by doing so-called ``internal opinion design",
knowing that the other adversarial player as the follower is to maximize the same objective by also conducting her own internal opinion design. We furthermore consider multiagent learning, specifically using the Optimistic Gradient Descent Ascent, and analyze its convergence to equilibria in the simultaneous version of competitive opinion optimization.

\end{abstract}

\section{Introduction}
The opinion forming process in a social network can be naturally thought as opinion influencing and updating dynamics.
This already attracted researchers' interest a while ago in mathematical sociology, and recently in theoretical computer science. DeGroot~\cite{degroot} modeled the opinion formation process by associating each individual with a numeric-value opinion and letting the opinion be updated by weighted averaging the opinions of her friends and her own, where the weights represent how much she is influenced by her friends. This update dynamics will converge to a consensus where all individuals hold the same opinions. However, we can easily observe that in the real world, consensus is difficult to be reached. Friedkin and Johnsen~\cite{friedkin} differentiated an expressed opinion that each individual in the networks updates over time from an internal opinion that each individual is born with and stays unchanged. Thus, an individual would always be influenced by her inherent belief, and the dynamics converges to an unique equilibrium, which may not be a consensus.
Bindel et al.~\cite{bindel} viewed the updating rule mentioned above equivalently as each player updating her expressed opinion to minimize her quadratic individual cost function, which consists of the disagreement between her expressed opinion and those of her friends, and the difference between her expressed and internal opinions. They analyzed how socially good or bad the system can be at equilibrium compared to the optimum solution in terms of the price of anarchy \cite{koutsoupias:papadimitriou:anarchy}. For directed graphs, we also had a price-of-anarchy result in a general class of graphs where no node is influencing others much more than being influenced \cite{chen:chen:lu}.

From a perspective of designing or engineering, \emph{opinion maximization (or minimization)} has been studied for seeding algorithms in \cite{gionis,ahmadinejad}.
With a \emph{linear} objective of the sum of expressed opinions,
opinion maximization seeks to find a $k$-subset (for a fixed size $k$) of nodes to have their expressed opinions fixed to 1 to maximize the objective.
Opinion minimization can be similarly defined to minimize the objective.
A seeding algorithm chooses what subset of nodes to fix their expressed opinions (to 1 if to maximize the objective),
and it turns out that opinion maximization is NP-hard \cite{gionis}. Thus, greedy algorithms \cite{gionis,ahmadinejad} have been designed to approximate the maximum with the help of the submodularity of such a social cost of the expressed opinion sum.

It is obvious to see that controlling the expressed opinions is not the only way to optimize the objective.
It is natural to consider changing the intrinsic (or equivalently, internal) opinions of some subset to optimize the objective.
Notice that setting a selected subset of nodes to have certain assigned internal opinions does not prohibit later deciding their expressed opinions by the influence and update dynamics while controlling the expressed opinions of the chosen subset is definitive.
In this sense, such an ``internal opinion designing" approach is relatively more relaxed, compared to the previously studied ``expressed opinion control"~\cite{gionis}.
What does not make much sense is a ``competitive version" of expressed opinion control where two players, one maximizing the objective and the other minimizing the same objective, forms a two-player zero-sum game by selecting their respective subsets of nodes to add influence values. Our \emph{internal opinion design} has a more meaningful competitive version detailed in the following subsection, where the optimal strategies for the min player and the max player to play are not anymore obvious.

\subsubsection{Our Results: Online and Multiagent Learning for Competitive Opinion Optimization.}
We then define the game of \emph{competitive opinion optimization} as follows.
One can think of a competitive scenario of two players, one with the goal to minimize (or maximize) the objective and the other adversarial player trying to do the opposite thing.
In such competitive opinion optimization, a zero-sum game is formed by these two players with all the possible combination of influence values added on nodes subject to a capacity constraint as the strategy set and each optimizing the same objective in the opposite direction.
The min player minimizes the sum of expressed opinions, i.e., doing internal opinion design discussed above, knowing that the max player as the follower is to maximize the same objective by also her own internal opinion design, and similarly for the max player.
Even if a node is influenced by a player for its intrinsic opinion design, its internal opinion would still be influenced by the other player.
Thus, a node's expressed opinion will be decided by its designed internal opinion (possibly by both players) and the update dynamic.

We first ask the problem of coming up with the min player's Stackelberg strategy against the max player's adversarial strategy as an online optimization problem in essence, specifically an online linear optimization one. From this idea, we design a randomized algorithm simply using the \emph{follow-the-perturbed-leader} algorithm~\cite{kalai:vempala} to produce candidate combinatorial strategies that are distributed according to the underlying probability distribution of mixed strategies and output a randomized combinatorial strategy at some uniformly chosen time step. In the literature of online optimization, the follow-the-perturbed-leader algorithm is a known approach to be adapted to tackle online combinatorial optimization problems, which we discuss more in the related work.
To uniformly randomize over the strategies at different time steps, the min player has to self-simulate playing follow-the-perturbed-leader, which requires obtaining every step's loss that depends on the other adversarial player's play. Thus, estimating the adversary's response strategy is also the min player's job since it is all in the min player's simulation. 
Then, we show that the strategy output by the randomized algorithm for the min player converges to an approximate min strategy against the other adversarial player, mainly taking advantage of the no-regret property in Section~\ref{sec:randomized}.
In other words, the min player using the proposed randomized algorithm to play such a Stackelberg game against the max adversarial player guarantees an approximate minimax equilibrium.
Note that the follow-the-perturbed-leader that we adopt uses uniformly random perturbation in each dimension, which is called the additive version, instead of the multiplicative version~\cite{kalai:vempala} or another version using Gumbel distributed perturbation. Follow-the-perturbed-leader using Gumbel distributed perturbation leads to essentially the exponential weights algorithm (or alternatively, multiplicative weights algorithm)~\cite{lee}. We did \emph{not} try to modify the follow-the-perturbed-leader in any way or devise any new combinatorial online learning algorithm, but simply simulated the play of the follow-the-perturbed-leader for our purpose of outputting the leader's (randomized) combinatorial strategy. The efficiency of our proposed randomized algorithm is inherently guaranteed and, moreover, its no-regret property can be used for our equilibrium strategy analysis.

Furthermore, we view our problem of coming up with the \emph{minimax} strategy and the \emph{maximin} strategy as \emph{multiagent extension of online learning/online convex optimization}.
Note that using generic or specific no-regret algorithms to play is also a common approach to reach certain equilibria \emph{on average} in repeated games~\cite[Chapter 4]{NRTV07}. 
Finally, we adapt the Optimistic Gradient Descent Ascent algorithm for the specific problem structure of competitive opinion optimization to derive the dynamics for both the min and max players and analyze the convergence to equilibria in the simultaneous version of competitive opinion optimization with a convergence rate in Section~\ref{sec:learning}.

\subsubsection{Relate Work} \label{sec:related}

\paragraph{Seed Selection Algorithms for Opinion Maximization.}
Using the sum of expressed opinions as the objective,
opinion maximization seeks to find a $k$-subset of nodes to have their expressed opinions fixed to 1 to maximize the objective.
Greedy algorithms have been designed to approximate the optimum with the help of the submodularity of such social cost~\cite{gionis,ahmadinejad}. We can view opinion maximization as a single-player problem compared with our competitive opinion optimization.

\paragraph{Connection to Combinatorial Online Optimization.} Online learning algorithms have been designed for making ``structured" decisions that are composed of components~\cite{koolen}, i.e., combinatorial online optimization problems. There could be an exponential number of decisions in terms of the number of components. To apply the well-known hedge algorithm~\cite{freund:schapire_1997,freund:schapire_1999} to the combinatorial expert setting, the experts are chosen as the structured decisions, which is called the extended hedge algorithm~\cite{koolen}. Obviously, one of the problems of this approach is to maintain exponentially many weights.
Learning with structured decisions has also been dealt with in the bandit setting where only the loss for the structured decision selected is available, which is called a combinatorial bandits problem~\cite{cesa-bianchi}.

\paragraph{No-Regret Play in Games.} It has been studied for two players playing no-regret algorithms to reach mixed Nash equilibrium (minmax equilibrium) in general zero-sum matrix-form games where the strategy set is finite~\cite{daskalakis:deckelbaum:kim}.

\paragraph{Competitive Influence Maximization.}
There are works on competitive versions of various (combinatorial) optimization problem other than competitive opinion optimization that we define in this paper.
The most well-known one is probably competitive influence maximization and its variation \cite{bharathi,goyal,he:kempe}.
Equilibrium computation and the analysis of the price-of-anarchy have been studied in these classes of games.

\section{Preliminaries} \label{sec:prelim}
Our game is based on the opinion formation game and its equilibrium.
First, we introduce the  fundamentals in opinion formation games. Then, we proceed with preliminaries of our competitive opinion optimization in Sect.~\ref{sec:game}. The performance measure for convergence will be introduced in Sect.~\ref{subsec:performance_measure}.

We describe a social network as a weighted graph $(G,\mathbf{w})$ for directed graph $G=(V,E)$ and weight matrix $\mathbf{w}=[w_{ij}]_{ij}$.
The node set $V$ of size $n$ represents the selfish players, and the edge set $E$ corresponds to the relationships between a pair of nodes.
The edge weight $w_{ij}\geq 0$ is a real number and represents how much player~$i$ is influenced by player~$j$;
note that weight $w_{ii}$ can be seen as a self-loop weight, i.e., how much player~$i$ influences (or is influenced by) herself.
Each (node) player has an internal opinion $s_i$, which is unchanged and not affected by opinion updates.
An opinion formation game can be expressed as an instance $(G,\mathbf{w},\mathbf{s})$ that combines weighted graph $(G,\mathbf{w})$ and vector $\mathbf{s}=(s_i)_i$.
Each player's strategy is an expressed opinion $z_i\in [-1,1]$, which may be different from her $s_i\in [-1,1]$ and gets updated.
Both $s_i$ and $z_i$ are real numbers.
The individual cost function of player~$i$ is
\begin{eqnarray*}
C_i(\mathbf{z})&=&w_{ii}(z_i-s_i)^2+\sum_{j\in N(i)}w_{ij}(z_i-z_j)^2
=w_{ii}(z_i-s_i)^2+\sum_j w_{ij}(z_i-z_j)^2,
\end{eqnarray*}
where $\mathbf{z}$ is the strategy profile/vector and $N(i)$ is the set of the neighbors of $i$, i.e., $\{j:j\neq i,w_{ij}>0\}$.
Each node minimizes her cost $C_i$ by choosing her expressed opinion $z_i$.
We analyze the game when it stabilizes, i.e., at equilibrium.

In a \emph{(pure) Nash equilibrium} $\mathbf{z}$, each player~$i$'s strategy is $z_i$ such that given $\mathbf{z}_{-i}$ (i.e., the opinion vector of all players except $i$) for any other $z'_i$,
\begin{equation} \label{eq:equil}
C_i(z_i,\mathbf{z}_{-i}) \leq C_i(z'_i,\mathbf{z}_{-i}).
\end{equation}
That is equivalently for each player to update her expressed opinion by the following rule \cite{bindel,bhawalkar}:
\begin{equation} \label{eq:update}
z_i=\frac{w_{ii}s_i+\sum_{j\neq i}w_{ij}z_j}{w_{ii}+\sum_{j\neq i}w_{ij}}.
\end{equation}
This is obtained by taking the derivative of $C_i$ w.r.t. $z_i$, setting it to $0$ for each $i$,
and solving the equality system since very player~$i$ minimizes $C_i$.
Note that $C_i$ is continuously differentiable.

In an opinion formation game, computing Nash equilibrium can be done by using \emph{absorbing random walks}~\cite{doyle:snell}.
In a random walk on a directed graph $H=(Z,R)$ with its weight matrix $W$,
a node in $Z$ is an absorbing node if the random walk can only enter this node but not exit from it,
and each entry $w_{i,j}\in R$ is the weight on edge~$(i,j)$.
Let $B\subseteq Z$ be the set of all \emph{absorbing} nodes, and the remaining nodes $U=Z\setminus B$ are \emph{transient} nodes.
Given the \emph{transition matrix} $P$ (from the weight matrix $W$) whose entry $P_{i,j}$ represents the probability transiting from node $i$ to node $j$  in this random walk, a $|U|\times|B|$ matrix $Q_{UB}$ can be computed
where each entry ${Q_{UB}}_{i,j}$ is \emph{the probability that a random walk starting at transient state $i\in U$ is absorbed at state $j\in B$} (see Appendix~\ref{app:a} for details).
If a random walk starting from transient node $i$ gets absorbed at an absorbing node $j$, we assign to node~$i$ the value~$b_j$ that is associated with absorbing node~$j$.
With $Q_{UB}$, the expected value of node~$i$ is then $f_i=\sum_{j\in B}{Q_{UB}}_{i,j}b_j$.
Let $f_U$ be the vector of the expected values for all $i\in U$ and $f_B$ the vector of values $b_j$ for all $j\in B$.
We have that
\begin{equation}
f_U=Q_{UB}f_B. \label{eq:absorbing}
\end{equation}

Thus, computing the expressed opinion vector at Nash equilibrium for an opinion formation game can be done by taking advantage of Equation~(\ref{eq:absorbing}) on a graph $H=(Z,R)$ constructed for our purpose as follows.
The weighted graph $(G,\mathbf{w})$ with original internal opinions $\mathbf{s}$ gives $U=V$ and $B=V'$ for the random walk on $H$,
where each $u_i\in V$ has a distinct copy $u'_i\in V'$ and $R=E\cup\{(u_i,u'_i):u_i\in V,u'_i\in V'\}$ with each weight $w_{u_i u'_i}=1$ and $f_B=\mathbf{s}$ so $\mathbf{z}=Q_{UB}\mathbf{s}$.
\begin{remark} \label{thm:nash}
The expressed opinion vector at Nash equilibrium can be computed in polynomial time in terms of the number of nodes.
\end{remark}

In the case of \emph{expressed opinion control} for opinion maximization as in \cite[Sect.~3.3]{gionis} by a seed selection algorithm, controlling the set $S\subseteq V$ gives $U=V\setminus S$ and $B=V'\cup S$ with all $b_j=1$ for $j\in S$, i.e., $f_B=(\mathbf{s},\mathbf{1})$, where $\mathbf{1}$ is a all-$1$ vector of size~$|S|$,
since the nodes in~$S$ cannot change their expressed opinions but stick to value~$1$. 
Our competitive \emph{internal opinion design} will be introduced in Sect.~\ref{sec:game}.
There, when using absorbing random walks for arriving at stable states,
$U=V$ and $B=V'$ along with the weighted edges remain as mentioned in the last paragraph yet with~$f_B$ being the internal opinion after manipulation.
Note that the matrix computations involved in computing $f_U$ (computing $Q_{UB}$ included) can done in polynomial time in terms of the number of states in the random walk by Remark~\ref{thm:nash}.

\subsection{Competitive Opinion Optimization} \label{sec:game}
A two-player \emph{competitive opinion optimization} can be described as an instance $((G,\mathbf{w}),\mathbf{s},\mathcal{X},\mathcal{Y},f)$.
We will elaborate each component one by one.
Knowing $\mathbf{y}$, let the min player's strategy be a vector $\mathbf{x}=(x_i)_i\in\mathbb{R}^n$ with $||\mathbf{x}||_1\leq k$.
Knowing $\mathbf{x}$, the max player's strategy is a vector $\mathbf{y}=(y_i)_i\in\mathbb{R}^n$ with $||\mathbf{y}||_1\leq k$.
Let $\mathcal{X}=\{\mathbf{x}\in \mathbb{R}^n:||\mathbf{x}||_1\leq k\}$ and $\mathcal{Y}=\{\mathbf{y}\in\mathbb{R}^n:||\mathbf{y}||_1\leq k\}$ denote the strategy sets for the min player and the max player, respectively.
Thus, for a node~$i$, after its internal opinion is affected by the min player and the max player, its internal opinion becomes \emph{a modified value~$s_i+x_i+y_i$ clipped between $-1$ and $1$}.
That is, for the min player, when $x_i<-1-a_i$ for $a_i=s_i+y_i$, the modified opinion $s_i+x_i+y_i$ stays at $-1$; for the max player, when $y_i>1-b_i$ for $b_i=s_i+x_i$, the modified opinion $s_i+x_i+y_i$ stays at $1$. 
Note that the expressed opinions are still influenced by $\mathbf{s}+\mathbf{x}+\mathbf{y}$ and get updated to the value at stable state by the dynamic,
using absorbing random walks (applying Equation~(\ref{eq:absorbing})).
We consider an objective $C(\mathbf{z})=\sum_i z_i$ that is the sum of all nodes' expressed opinions $z_i$.

The min player minimizes her cost function over all $\mathbf{x}$'s, which the max player maximizes,
\begin{eqnarray} \label{eq:obj}
f(\mathbf{x},\mathbf{y})=C(Q_{UB}(\mathbf{s}+\mathbf{x}+\mathbf{y}))=\ell^\mathbf{T}(\mathbf{s}+\mathbf{x}+\mathbf{y})
\end{eqnarray}
for $U=V$ and $B=V'$ and a vector $\ell=(\sum_i {Q_{UB}}_{i,j})_j$.

\paragraph{Online Linear/Convex Optimization.} \label{sec:olo}
In the setting of online convex optimization, we describe an online game between a player and the environment.
The player is given a convex set $\mathcal{K}\subset \mathbb{R}^d$ and has to make a sequence of decisions $\mathbf{x}^{(1)},\mathbf{x}^{(2)},\ldots \in\mathcal{K}$.
After deciding $\mathbf{x}^{(t)}$, the environment reveals a convex reward function~$f^{(t)}$ and the player obtains $f^{(t)}(\mathbf{x}^{(t)})$.
Which is closely related to our problem is a more specific problem of online linear optimization where the reward functions are linear, i.e., $f^{(t)}(\mathbf{x})=\langle F^{(t)},\mathbf{x}\rangle$ for some $f^{(t)}\in \mathbb{R}^d$.

We define the player's adaptive strategy $\mathcal{L}$ as a function taking as input a subsequence of loss vectors $F^{(1)},\ldots,F^{(t-1)}$ and returns a point $\mathbf{x}^{(t)}\leftarrow\mathcal{L}(F^{(1)},\ldots,F^{(t-1)})$ where $\mathbf{x}^{(t)}\in\mathcal{K}$.
The performance of the player is measured by \emph{regret} defined in the following.
\begin{definition}
Given an online linear optimization algorithm $\mathcal{L}$ and a sequence of loss vectors $F^{(1)},F^{(2)},\ldots\in \mathbb{R}^n$,
let the regret $\operatorname{Regret}(\mathcal{L};F_{1:T})$ be defined as
\[\sum_{t=1}^T\langle F^{(t)},\mathbf{x}^{(t)}\rangle-\min_{\mathbf{x}\in\mathcal{K}}\sum_{t=1}^T\langle F^{(t)},\mathbf{x}\rangle.
\footnote{For a player maximizing her total reward given a sequence of reward vectors, the regret can also be defined accordingly.}\]
\end{definition}

A desirable property that one would want an online linear optimization algorithm to have is a regret which scales sublinearly in~$T$.
For example, the online gradient descent algorithm \cite{zinkevich} guarantees a regret of~$O(\sqrt{T})$.
This property can be formally captured as the following.
\begin{theorem}[e.g., Theorem 10 of \cite{abernethy:bartlett}] \label{thm:no-regret}
For any bounded decision set $\mathcal{K}\subseteq \mathbb{R}^d$ there exists an algorithm $\mathcal{L}_\mathcal{K}$ such that $\operatorname{Regret}(\mathcal{L}_\mathcal{K})=o(T)$ for any sequence of loss vectors $\{F^{(t)}\}$ with bounded norm.
\end{theorem}

The no-regret property above is useful in a variety of contexts.
For example, it is known (e.g., \cite[Section 3]{abernethy:bartlett}) that two players playing $o(T)$-regret algorithms $\mathcal{L}_\mathcal{X}$ and $\mathcal{L}_\mathcal{Y}$, respectively, in a zero-sum game with a cost function $f : \mathcal{X}\times \mathcal{Y}\rightarrow R$ of the form $f(\mathbf{x},\mathbf{y}) = \mathbf{x}^T M\mathbf{y}$ for some $M\in \mathbb{R}^{n\times m}$ give a version of minmax equilibrium whose proof is restated in Appendix~\ref{app:b}.
\begin{theorem}[Corollary 3 of \cite{abernethy:bartlett}] \label{thm:minmax}
For compact convex sets $\mathcal{X} \subset \mathbb{R}^n$ and $\mathcal{Y} \subset \mathbb{R}^m$ and any biaffine function\footnote{A biaffine function $f : \mathcal{X} \times \mathcal{Y} \rightarrow R$ satisfies $f(\alpha \mathbf{x} + (1-\alpha)\mathbf{x}',\mathbf{y}) = \alpha f(\mathbf{x},\mathbf{y}) + (1-\alpha)f(\mathbf{x}',\mathbf{y})$ and
$f(\mathbf{x},\alpha \mathbf{y} + (1-\alpha)\mathbf{y}') = \alpha f(\mathbf{x},\mathbf{y}) + (1-\alpha)f(\mathbf{x},\mathbf{y}')$ for every $0 \leq \alpha \leq 1$, $\mathbf{x},\mathbf{x}'\in \mathcal{X}$ and $\mathbf{y},\mathbf{y}'\in \mathcal{Y}$.}
$f : \mathcal{X} \times \mathcal{Y} \rightarrow \mathbb{R}$, we have
\begin{eqnarray}
\min_{\mathbf{x}\in \mathcal{X}}\max_{\mathbf{y}\in \mathcal{Y}}f(\mathbf{x},\mathbf{y})=\max_{\mathbf{y}\in \mathcal{Y}}\min_{\mathbf{x}\in \mathcal{X}}f(\mathbf{x},\mathbf{y}).
\end{eqnarray}
\end{theorem}

This standard technique and result have been existing for playing generic no-regret algorithms in a zero-sum $n\times m$ matrix game.
One can view the argument as something we would like to do on a high level but with different technical details for coping with our game.

For our competitive opinion optimization, one can first notice that the strategies of the two players interact with each other and matrix $Q_{UB}$, which corresponds to the cost matrix $M$, in a very different way from the standard result discussed in Theorem~\ref{thm:minmax}.
For example, we have $Q_{UB}(\mathbf{s}+\mathbf{x}+\mathbf{y})$ here instead of $\mathbf{x}^T M \mathbf{y}$.
Given the strategy sets $\mathcal{X}$ and $\mathcal{Y}$,
our competitive opinion optimization takes a cost function $f: \mathcal{Z}=\mathcal{X}\times \mathcal{Y}\rightarrow R$ defined in (\ref{eq:obj}).
The min player imagines engaging herself in an online optimization problem and the max player in the other online optimization problem, where at time $t$ the min player chooses $\mathbf{x}^{(t)}$ and the max player chooses $\mathbf{y}^{(t)}$.
In one online optimization, \emph{the min player chooses her strategies according to the play of a \emph{no-regret} algorithm, assuming that the max player in response maximizes the value of the objective each time step, and in the other online optimization, the max player does similarly.} These two directions of optimization can actually be done jointly and simultaneously, which will be detailed in Section~\ref{sec:learning}.

However, there are subtleties that need to be coped with. One can view Theorem~\ref{thm:minmax} (Corollary 3 of \cite{abernethy:bartlett}) as something we would like to do on a high level \emph{approximately} but with different technical details for coping with our competitive opinion optimization since we have $f(\mathbf{x},\mathbf{y})=\ell^\mathbf{T}(\mathbf{s}+\mathbf{x}+\mathbf{y})$. In particular, for a node~$i$, after its internal opinion is affected by the min player and the max player, \emph{without} clipping its modified internal opinion value~$s_i+x_i+y_i$ clipped to the range $[-1,1]$ every node's equilibrium strategy would result in a \emph{dominant strategy} solution, which is a special case of Nash equilibria and a less interesting target to look for since the strategies of the min player and the max player would not be mutually entangled.

\subsection{Performance Measure for Convergence}
\label{subsec:performance_measure}

For the notation $\mathcal{Z}$, which is a concatenation of~$\mathcal{X}$ and~$\mathcal{Y}$, and accordingly $\mathcal{Z}^{*}$ is also a concatenation of~$\mathcal{X}^{*}$ and~$\mathcal{Y}^{*}$ only for convenience. 
For a point, $(\mathbf{x}, \mathbf{y})$ as a strategy profile can be directly represented by $\mathbf{z}\in\mathcal{Z}$, and define $F(\mathbf{z})=\left(\nabla_{\mathbf{x}} f(\mathbf{x}, \mathbf{y}),-\nabla_{\mathbf{y}} f(\mathbf{x}, \mathbf{y})\right)$.

In this minimax equilibrium game, there are many ways to measure the convergence performance, and some unique distances are often used to estimate the convergence rate. 
The duality gap, defined as $\alpha_{f}(\mathbf{z})=\max _{\mathbf{y}^{\prime} \in \mathcal{Y}} f\left(\mathbf{x}, \mathbf{y}^{\prime}\right)-\min _{\mathbf{x}^{\prime} \in \mathcal{X}} f\left(\mathbf{x}^{\prime}, \mathbf{y}\right)$, which is always a positive value. This term has been used in many works, either for theorem proof or numerical experiments. The notation $\operatorname{dist}\left(\mathbf{z},\mathcal{Z}^{*}\right)$ is the squared distance between $\mathbf{z}$ and $\mathcal{Z}^{*}$, which can be formulated as $\left\|\mathbf{z}-\Pi_{\mathcal{Z}^*}(\mathbf{z})\right\|^{2}$, and we denote by~$\Pi_{\mathcal{Z}^*}(\mathbf{z}) := \operatorname{argmin}_{\mathbf{z}\in\mathcal{Z}^*} \operatorname{dist}\left(\mathbf{z}, \mathcal{Z}^{*}\right)$ the projection of~$\mathbf{z}$ onto~$\mathcal{Z}^*$. 

\section{Randomized Algorithm for the Leader's Strategy} \label{sec:randomized}
For our opinion optimization game, one can first notice that the strategies of the two players interact with each other and matrix $Q_{UB}$, which corresponds to the cost matrix~$M$, in a very different way from the standard result discussed in Section~\ref{sec:olo}.
For example, we have $Q_{UB}(\mathbf{s}+\mathbf{x}+\mathbf{y})$ here instead of $\mathbf{x}^T M \mathbf{y}$.
Being aware of the differences, in this section we henceforth design algorithms for computing an approximate equilibrium strategy of the min player (against the adversarial player),
and focus on efficient computation of the adversary's strategy as well as the equilibrium strategy analysis only for the min player, instead of characterizing equilibrium, i.e., equilibrium strategies for both players
(since the adversarial player can overwrite the min player's selection and we do not have a symmetric structure such as $\mathbf{x}^T M \mathbf{y}$ in our problem).

Given the strategy sets $\mathcal{X}$ and $\mathcal{Y}$,
our Stackelberg opinion optimization game takes a cost function $f: \mathcal{X}\times \mathcal{Y}\rightarrow \mathbb{R}$ defined in Section~\ref{sec:game}.
The min player imagines engaging herself in an online optimization problem, where at time $t$ the min player chooses $\mathbf{x}^{(t)}$ and the adversarial (max) player chooses $\mathbf{y}^{(t)}$.
In such online optimization, \emph{the min player chooses her strategies according to simulating the play of a \emph{no-regret} algorithm, assuming that the adversarial player in response maximizes the value of the objective each time step.} Also, the min player can select a strategy at some time $T_{\min}$ chosen uniformly at random, and this randomized strategy will be shown to be an approximate equilibrium strategy.

\subsection{Simulating the Play of the Follow-the-Perturbed-Leader Algorithm}
Specifically, transforming our problem of finding the leader's strategy into an online linear optimization, we simulate playing the additive version of the follow-the-perturbed-leader algorithm \cite{kalai:vempala} 
to obtain a sequence of ``randomized" (fractional) strategies; we then get an ``average" (over time steps) randomized (fractional) strategy as the output of our randomized algorithm.

For every time step $t$, the (fractional) strategy of the leader $\mathbf{x}^{(t)}\leftarrow\mathcal{L}_\mathcal{X}(f^{(1)}(\cdot),\ldots,f^{(t-1)}(\cdot))$ can be rounded into the integral strategy~$\mathbf{x}^{(t)}$, noting $\mathcal{X}=\{\mathbf{x}\in \mathbb{R}^n: \|\mathbf{x}\|_0\leq k, 0\leq\bar{x}_i\leq 1\}$ and $\mathcal{L}_\mathcal{X}$ is the additive version of the follow-the-perturbed-leader algorithm, and estimating
\[\mathbf{y}^{(t)}\simeq\arg\max_{\mathbf{y}\in \mathcal{Y}}f(\mathbf{E}[\mathbf{x}^{(t)}],\mathbf{y}).\footnote{We use the notation $\mathbf{E}[\mathbf{x}]$ to denote an expected vector throughout this paper.}\]

\paragraph{Algorithm Design.} We are now ready to specify simulating the play of the additive version of the follow-the-perturbed-leader algorithm plus the rounding for getting a randomized combinatorial strategy $\mathbf{x}^{(t)}$ at each time step $t$: the min player's fractional strategy at time step $t$ is
\begin{eqnarray*}
&&\mathbf{x}^{(t)}
=\arg\min_{\mathbf{x}\in\mathcal{X}}L^{(t-1)}(\mathbf{x})
=\arg\min_{\mathbf{x}\in\mathcal{X}}(\sum_{\tau=1}^{t-1}f^{(\tau)}(\mathbf{x})+R_t\mathbf{x}),
\end{eqnarray*}
for a random vector $R_t\in[0,\sqrt{T}]^{n}$ uniformly distributed in each dimension.
Since $L^{(t-1)}$ is affine in $\mathbf{x}\in\mathcal{X}$ and the constraints forming $\mathcal{X}$ are linear as well, the minimizer $\mathbf{x}^{(t)}$ can be computed efficiently.
Actually, $\textbf{E}_{\mathbf{x}^{(t)}\sim X^{(t)}}[\mathbf{x}^{(t)}]$ can be estimated by sampling $\mathbf{x}^{(t)}$ enough times,
which we will use and explain in Appendix~\ref{app:g}.
Thus, we conclude that our randomized algorithm outputs a (randomized) pure strategy in a uniformly random time step $T_{\min}$ for the min player against the adversarial player who ideally is to play $\arg\max_{\mathbf{y}\in \mathcal{Y}}g(\mathbf{E}_{\mathbf{x}^{(t)}\sim X^{(t)}}[\mathbf{x}^{(t)}],\mathbf{y})$ at each time step $t$.
\begin{proposition} \label{pro:adversary}
We can estimate this strategy of the adversary as accurately as possible with high probability.
That is, with high probability
\[f(\mathbf{E}_{\mathbf{x}^{(t)}\sim \mathbf{x}^{(t)}}[\mathbf{x}^{(t)}],\mathbf{y}^{(t)})\geq\arg\max_{\mathbf{y}\in \mathcal{Y}}f(\mathbf{E}_{\mathbf{x}^{(t)}\sim \mathbf{x}^{(t)}}[\mathbf{x}^{(t)}],\mathbf{y})-\epsilon,\]
where $\epsilon>0$ is an error from estimation, which can be made as small as desired.
\end{proposition}
We show that such strategy $\mathbf{y}^{(t)}$ of the adversary can be found efficiently in the Appendix~\ref{app:g}.
The randomized algorithm simulates the follow-the-perturbed-leader algorithm up to time step $T_{\min}$.
Our main result is to show that the randomized pure strategy indeed approaches an approximate minimax equilibrium strategy (see Section~\ref{sec:equilibrium}).
The randomized algorithm is summarized as follows.

\begin{algorithm} \label{alg:1}
\caption{Randomized algorithm for the leader's strategy}
\begin{algorithmic}[1]
\STATE Choose $T_{\min}$ uniformly at random from $\{1,\ldots,T\}$
\FOR{$t=1$ to $T_{\min}$}
\STATE $\mathbf{x}^{(t)}=\arg\min_{\mathbf{x}\in\mathcal{X}}(\sum_{\tau=1}^{t-1}f^{(\tau)}(\mathbf{x})+R_t\mathbf{x})$ for a uniformly random (in each dimension) vector $R_t\backsim U[0,\sqrt{T}]^n$,
where the adversary's $\mathbf{y}^{(\tau)}$ 
can be efficiently computed. 
\STATE Estimating the adversary's strategy $\mathbf{y}^{(t)}$ 
to achieve Proposition~\ref{pro:adversary}.
\ENDFOR

\end{algorithmic}
\end{algorithm}

\subsection{Equilibrium Strategy Analysis} \label{sec:equilibrium}
Let the min player play the strategy output by the randomized algorithm and
the adversarial player's strategy be the one maximizing the loss, given the min player's chosen strategy.
First, it can be shown that the play output by the randomized algorithm is ``nearly" $O(\frac{n^{3/2}}{\sqrt{T}})$-average regret with high probability (see the proof of Lemma~\ref{eq:no-regret}).
This is achieved naturally in the sense of expected losses of the min player since there is a random vector $R_t$ as a random source that produces the distribution $\mathbf{x}^{(t)}$. The proof is deffered to Appendix~\ref{app:h}.

\begin{lemma} \label{eq:no-regret}
For the min player, the follow-the-perturbed-leader algorithm is nearly $\frac{2(1+\delta)n^{3/2}}{\sqrt{T}}$-average regret w.r.t. her respective loss functions depending on the adversary's strategy $\mathbf{y}^{(t)}$'s.
\end{lemma}

Since the randomized algorithm chooses time step $T_{\min}$ uniformly at random from $1,\ldots,T$, we let
\[\mathbf{E}_{T_{\min}\in\{1,\ldots,T\},\mathbf{x}^{(T_{\min})}\sim X^{(T_{\min})}}[\mathbf{x}^{(T_{\min})}]
=\frac{\sum_{t=1}^T\mathbf{E}_{\mathbf{x}^{(t)}\sim \mathbf{x}^{(t)}}[\mathbf{x}^{(t)}]}{T}.\]
Then, we are ready to state the main result whose proof is detailed in Appendix~\ref{app:c}.
\begin{theorem} \label{thm:strategy}
The strategy $\mathbf{x}^{(T_{\min})}$ output by the randomized algorithm for the min player against the adversarial player 
is a $(1+\delta,\frac{\sqrt{\ln T}+2(1+\delta)n^{3/2}}{\sqrt{T}})$-approximate equilibrium strategy for $(1+\delta)$ multiplicative approximation and $(\frac{\sqrt{\ln T}+2(1+\delta)n^{3/2}}{\sqrt{T}})$ additive approximation with probability of at least $1-\frac{2}{T}$ for some constant $0<\delta\leq 1$.
\end{theorem}

\section{Multiagent Learning: Optimistic Mirror Descent Ascent for Simultaneous Competitive Opinion Optimization} \label{sec:learning}

Combining the play of a specific no-regret algorithm, say \emph{the optimistic mirror descent, for the min player} and that of a specific no-regret algorithm, say \emph{the optimistic mirror ascent, for the max player}, we propose to use the Optimistic Mirror Descent Ascent (OMDA) algorithm \cite{chiang:yang,hsieh,golowich,wei:lee:zhang:luo}. 
Note that in the literature of online learning, the Optimistic Weights Update (OMWU) algorithm and Optimistic Gradient Dsecent Ascent (OGDA) algorithm can be viewed as special cases of Optimistic Mirror Descent Ascent. Specifically, let the negative entropy~$\sum_i u_i\ln u_i$ be the regularizer~$\mathcal{R}(\mathbf{u})$ for the case of OMWU and half of the $l_2$ norm square ~$\frac{1}{2}\|\mathbf{u}\|^2_2$ for the case of OGDA so that the Bregman divergence~$\mathcal{B}^\mathcal{R}(\mathbf{u},\mathbf{v})$ is KL-divergence~$KL(\mathbf{u},\mathbf{v})$ and $\frac{1}{2}\|\mathbf{u}-\mathbf{v}\|^2_2$, respectively.

We adapt the OGDA \cite{wei:lee:zhang:luo}, which guarantees the last-iterate convergence, including bimatrix and convex-concave settings. OGDA plays a crucial role in computing $\mathbf{x}_{t}, \hat{\mathbf{x}}_{t+1}$ via gradient of $f\left(\mathbf{x}_{t-1}, \mathbf{y}_{t-1}\right)$ and $f\left(\mathbf{x}_{t}, \mathbf{y}_{t}\right)$ (and similarly $\mathbf{y}_{t}, \hat{\mathbf{y}}_{t+1}$ , using gradient of $f_{\mathbf{y}}\left( \mathbf{x}_{t-1}, \mathbf{y}_{t-1}\right)$ and $f_{\mathbf{y}}\left(\mathbf{x}_{t}, \mathbf{y}_{t}\right)$). 
At the end of each iteration~$t$, the objective function $\sum_{i}(\sum_{j} Q_{U B}(i, j) \cdot \textbf{clip}_j(s_{j} +x(j)+y(j)))=(\sum_{j} (\sum_{i}Q_{U B}(i, j)) \cdot \textbf{clip}_j(s_{j} +x(j)+y(j)))$ is updated to obtain the convergence to the minimax equilibrium with a convergence rate in our competitive opinion optimization.

\begin{algorithm}[ht]
  \caption{Optimistic Gradient Descent Ascent for Competitive Opinion Optimization}
    \begin{algorithmic}[1]
    \STATE $\mathbf{Parameters:}$ $i$ (index of a node), $\eta>0$, vector $\mathbf{s}=\left(s_{i}\right)_{i},|i| \times|j|$ matrix $Q_{U B}$
    \STATE $\mathbf{Initialization:}$ $$\begin{gathered}\mathbf{x}_{0}, \hat{\mathbf{x}}_{1},\left(x_{i}\right)_{i} \in \mathbb{R}^{n},\|\mathbf{x}\|_{1} \leq k, \\
    \mathbf{y}_{0}, \widehat{\mathbf{y}}_{1},\left(y_{i}\right)_{i} \in \mathbb{R}^{n},\|\mathbf{y}\|_{1} \leq k.\end{gathered}$$

    \FOR{$t = 1, \ldots, T$}
      \STATE if (GDA): \quad
      $\widehat{\mathbf{x}_{t+1}} = \mathbf{x}_{t},\quad
      \widehat{\mathbf{y}_{t+1}} = \mathbf{y}_{t} ,
      $
      \STATE Update the min player's strategy and the max player's strategy:
        $$
        \begin{aligned}
        &\mathbf{x}_{t}=\prod_{\mathbf{x}}\left(\widehat{\mathbf{x}_{t}}-\eta \nabla_{\mathbf{x}} \bar{f}\left(\mathbf{x}_{t-1}, \mathbf{y}_{t-1}\right)\right), 
        &\widehat{\mathbf{x}_{t+1}}=\prod_{\mathbf{x}}\left(\widehat{\mathbf{x}_{t}}-\eta \nabla_{\mathbf{x}} \bar{f}\left(\mathbf{x}_{t}, \mathbf{y}_{t}\right)\right),\\
        &\mathbf{y}_{t}=\prod_{\mathbf{y}}\left(\widehat{\mathbf{y}_{t}}+\eta \nabla_{\mathbf{y}} \bar{f}\left(\mathbf{x}_{t-1}, \mathbf{y}_{t-1}\right)\right), 
        &\widehat{\mathbf{y}_{t+1}}=\prod_{\mathbf{y}}\left(\widehat{\mathbf{y}_{t}}+\eta \nabla_{\mathbf{y}} \bar{f}\left(\mathbf{x}_{t}, \mathbf{y}_{t}\right)\right).
        \end{aligned}
        $$
       \STATE The gradient vector of $\bar{f}$ with respect to the min player's strategy is computed as follows:
        $$
        \left(\left\{\begin{array}{cl}
        \sum_i Q_{U B}(i, j)=\ell_j, & \text { if } \mathbf{x}_{t}(j) > -1-\left(s_{j}+\mathbf{y}_{t}(j)\right) \\
        \text{non-differentiable},& \text { if } \mathbf{x}_{t}(j) = -1-\left(s_{j}+\mathbf{y}_{t}(j)\right)\\
        0, & \text{ if } \mathbf{x}_{t}(j) < -1-\left(s_{j}+\mathbf{y}_{t}(j)\right)
        \end{array}\right.\right),
        $$ where $\beta \cdot \sum_i Q_{U B}(i, j)$ for any $\beta\in\left[ 0,1 \right]$ is a subgradient with respect to node~$j$'s strategy if $\mathbf{x}_{t}(j) = -1-\left(s_{j}+\mathbf{y}_{t}(j)\right)$.\\
       And that w.r.t. the max player's is computed similarly.
    \ENDFOR
  \end{algorithmic}
  \label{alg:algo1}
\end{algorithm}


\subsection{Convergence Results for OGDA}
Therefore, using the objective function of each dimension ($\ell_{i} \cdot \textbf{clip}_{i}\left(s_{i}+x_{i}+y_{i}\right)$), we want to approximate the original three-piecewise functions to strongly convex and strongly concave functions. The following proposition makes the approximation between $f_{\mathbf{y}}(x)$ (original three-piecewise linear function) and $\bar{f}_{\mathbf{y}}(\mathbf{x})$ (two-piecewise linear function), from the coordinate diagram, in which the x-axis is the value of $x_i$, and the y-axis is the value of $\ell_{i} \cdot \textbf{clip}_{i}\left(s_{i}+x_{i}+y_{i}\right)$, we can obtain the distance bound between the two functions at most $k$ (capacity limit) which is much smaller than $n$ whose proof is in Appendix~\ref{app:d}. 

\begin{proposition}[$\bar{f}$ as an approximation of $f$] \label{thm:approx}
Given $y$, we have $ \bar{f}_{\mathbf{y}}(\mathbf{x})-f_{\mathbf{y}}(\mathbf{x}) \leq k$ for all $\mathbf{x}$.
\end{proposition}
\begin{figure} [h]
\centering  
\includegraphics[height=5.5cm,width=7cm]{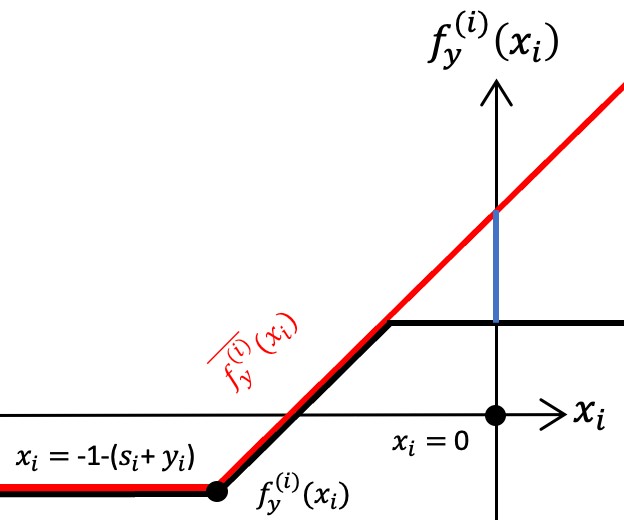}  
\caption{Diagram of ${f}_{\mathbf{y}}^{(i)}(x_{i})$ and $\bar{f}_{\mathbf{y}}^{(i)}(x_{i})$}  
\label{fig:figure3.1}  
\end{figure} 

The next proposition makes the approximation between $\tilde{f}_{\mathbf{y}}(\mathbf{x})$ (convex function) and $\bar{f}_{\mathbf{y}}(\mathbf{x})$ (two-piecewise linear function), we characterize the distance between the two functions in Proposition~\ref{thm:distance}.

\begin{proposition}[$\tilde{f}$ as an approximation of $\bar{f}$] \label{thm:distance}
For $x$ with $\|x\|_{1} \leq k$, we have 
\begin{eqnarray*}
&&\tilde{f}_y(x)-\bar{f}_y(x)=\sum_i\tilde{f}_{y}^{(i)}(x_i)-\sum_i\bar{f}^{(i)}_y(x_i)\\
&\leq&\sum_i\bar{f}_y\left(x_q(i)\right)-\big(\tilde{f}_{y}^{(i)}\left(x_p(i)\right)+\langle\nabla\tilde{f}_{y}^{(i)}\left(x_p(i)\right), x_q(i)-x_p(i)\rangle\big),
\end{eqnarray*}
 where $x_q(i)=-1-(s_{i}+y_{i})$ and $x_p(i)$ is a minimizer of $\tilde{f}_{y}^{(i)}(x_{i})$. And for $y$ with $\|y\|_1\leq k$, $\sum_{i}\tilde{f}_{x}^{(i)}(y_p(i))+\langle\nabla\tilde{f}_{x}^{(i)}\left(y_p(i)\right),y_q(i)-y_p(i) \rangle -\tilde{f}_{x}^{(i)}(y_q(i))$ represents the maximum overall distance between $\bar{f}_{{x}}(y)$ and $\tilde{f}_{{x}}(y)$
, where $y_q(i)=1-(s_{i}+x_{i})$ and $y_p(i)$ is a maximizer of $\tilde{f}_{x}^{(i)}(y_{i})$. 
\end{proposition}

Given a differentiable convex function $\tilde{f}_{\mathbf{y}}^{(i)}(x_{i})$, we assume that the linear approximation of $\tilde{f}_{\mathbf{y}}^{(i)}(x_{i})$ is the line that lies flat in $\bar{f}_{\mathbf{y}}^{(i)}(x_{i})$. We can find in the range of ${x}_i$ that the largest distance between $\bar{f}_{\mathbf{y}}^{(i)}(x_{i})$ and $\tilde{f}_{\mathbf{y}}^{(i)}(x_{i})$ is its turning point, $-1-(s_{i}+y_{i})$ in our game, to the intersection of its vertical line with $\tilde{f}_{\mathbf{y}}^{(i)}(x_{i})$
, which is exactly the Bregman Divergence by definition. 




We further derive the following proposition proved in Appendix~\ref{app:e} for the strong convexity of $\tilde{f}(\mathbf{x}^{\prime},\mathbf{y})$ and also the strong concavity of $\tilde{f}(\mathbf{x},\mathbf{y}^{\prime})$.

\begin{proposition}[Strong convexity of $\tilde{f}$, fixing $y$] \label{thm:strong_conv}
We have

$$
\tilde{f}(\mathbf{x}, \mathbf{y})-\tilde{f}\left(\mathbf{x}^{*}, \mathbf{y}\right) \leq \nabla_{\mathbf{x}} \tilde{f}(\mathbf{x},\mathbf{y})^{\top}\left(\mathbf{x}-\mathbf{x}^{*}\right)-\gamma\left\|\mathbf{x}-\mathbf{x}^{*}\right\|^{2} / 2
$$
for all~$\mathbf{x}$, where $\gamma_{i}$ satisfies for all~$i$
$$
\tilde{f}(\mathbf{x}, \mathbf{y})-\tilde{f}\left(x_{i}^{*}, \mathbf{x}_{-i}, \mathbf{y}\right) \leq \nabla_{x_{i}} \tilde{f}(\mathbf{x}, \mathbf{y})^{\top}\left(x_{i}-x_{i}^{*}\right)-\gamma_{i}|x_{i}-x_{i}^{*}|^{2} / 2.
$$
\end{proposition}

We further use these propositions to derive the convergence result that copes with our game in the following subsection. First, the average-iterate convergence is deferred completely to Appendix~\ref{app:f} due to its standard technique. Then, we elaborate on the last-iterate convergence.

\subsubsection{Last-Iterate Convergence}
Here, with a further approximation~$\tilde{f}$ of an proper approximation~$\bar{f}$ of $f$ by Proposition~\ref{thm:distance} and \ref{thm:approx}, we can further derive results that are similar to those in \cite{wei:lee:zhang:luo} but specific to the technical details of the competitive opinion optimization game.

We make the assumption that $f$ is $L$-smooth, which is in place with $\left\|F(\mathbf{z})-F\left(\mathbf{z}^{\prime}\right)\right\| \leq L\left\|\mathbf{z}-\mathbf{z}^{\prime}\right\|$ and also $\left\|\mathbf{z}-\mathbf{z}^{\prime}\right\| \leq 1$ for any $\mathbf{z}, \mathbf{z}^{\prime} \in \mathcal{Z}$, and recall \emph{SP-RSI-2} (Definition 1 of \cite{wei:lee:zhang:luo}) here, which is a general condition according to $f(\mathbf{x}, \mathbf{y})$ and $\mathcal{Z}$ to obtain the results of last-iterate convergence.

\begin{definition}[Definition 1 (Generalized Saddle-Point Restricted Secant Inequality (SP-RSI)) of \cite{wei:lee:zhang:luo}] \label{def:sp-rsi}
Condition \emph{SP-RSI-2} is defined as: for any $\mathbf{z} \in \mathcal{Z}$ with $\mathbf{z}^{*}=\Pi_{\mathcal{Z}^{*}}(\mathbf{z})$, where a point $\mathbf{z}=(\mathbf{x}, \mathbf{y}) \in \mathcal{Z}$, and $\mathcal{Z}=\mathcal{X} \times \mathcal{Y}$,
$$
\begin{aligned}
&\text { (SP-RSI-2) }  \quad\qquad F(\mathbf{z})^{\top}\left(\mathbf{z}-\mathbf{z}^{*}\right) \geq C\left\|\mathbf{z}-\mathbf{z}^{*}\right\|^{\beta+2}
\end{aligned}
$$
holds for some parameters $\beta \geq 0$ and $C>0$.
\end{definition}

\begin{lemma}[Theorem~6 of \cite{wei:lee:zhang:luo}]
\label{lemma:2}
If ${f}$ is strongly convex in $\mathbf{x}$ and strongly concave in $\mathbf{y}$, then SP-RSI-2 holds with $C=\frac{\gamma}{2}$ and $\beta=0$.
\end{lemma}


Under the SP-RSI-2 condition with a value of $\beta=0$, we can have a last-iterate convergence guarantee which is analogous to Theorem~8 of \cite{wei:lee:zhang:luo}.
\begin{lemma}[Theorem~8 of \cite{wei:lee:zhang:luo}]
\label{lemma:3}
For any $\eta \leq \frac{1}{8 L}$, if SP-RSI-2 holds with $\beta=0$, then OGDA guarantees linear last-iterate convergence:
$$
\operatorname{dist}\left(\mathbf{z}_{t}, \mathcal{Z}^{*}\right) \leq 96\left(1+C_{5}\right)^{-t}
$$
where $C_{5} \triangleq \frac{15 \min \left\{\eta^{2} C^{2}, 1\right\}}{81(1+\beta) \cdot 2^{\beta}}$.
\end{lemma}

Based on the existing convergence result of Lemma \ref{lemma:3}, we can also get a convergence guarantee for $\alpha_{\tilde{f}}\left(\mathbf{z}_{t}\right)$ (duality gap of $\mathbf{z}_{t}$) when it is a Lipschitz continuous function $f$, which is because 

\begin{eqnarray} \label{eq:duality}
\alpha_{\tilde{f}}\left(\mathbf{z}_{t}\right) 
&\leq \max _{\mathbf{x}^{\prime}, \mathbf{y}^{\prime}}\tilde{f}\left(\mathbf{x}_{t}, \mathbf{y}^{\prime}\right)-\tilde{f}\left(\mathbf{x}^{*}, \mathbf{y}^{\prime}\right)+ \tilde{f}\left(\mathbf{x}^{\prime}, \mathbf{y}^{*}\right)-\tilde{f}\left(\mathbf{x}^{\prime}, \mathbf{y}_{t}\right) \nonumber\\ 
&\leq \mathcal{O}\left(\left\|\mathbf{x}_{t}-\mathbf{x}^{*}\right\|+\left\|\mathbf{y}_{t}-\mathbf{y}^{*}\right\|\right)
=\mathcal{O}\left(\sqrt{\operatorname{dist}\left(\mathbf{z}_{t}, \mathcal{Z}^{*}\right)}\right),
\end{eqnarray} where $\left(\mathbf{x}^{*}, \mathbf{y}^{*}\right)=\Pi_{\mathcal{Z}^{*}}(\mathbf{z}_{t})$.

\begin{theorem} \label{thm:rate}
Algorithm \ref{alg:algo1} with $\eta \leq \frac{1}{8 L}$ where SP-RSI-2 only holds with $\beta=0$ guarantees a linear last-iterate convergence with certain approximation errors:
\begin{eqnarray*}
\alpha_{f}\left(z_{t}\right) &\leq& \mathcal{O}\left(\sqrt{\operatorname{dist}\left(\mathbf{z}_{t}, \mathcal{Z}^{*}\right)}\right) \\&&+\sum_{i}\tilde{f}_{\mathbf{y}}^{(i)}\left((x_{i})_{q}\right)-\big(\tilde{f}_{\mathbf{y}}^{(i)}\left((x_{i})_{p}\right)+\langle\nabla\tilde{f}_{\mathbf{y}}^{(i)}\left((x_{i})_{p}\right), (x_{i})_{q}-(x_{i})_{p}\rangle\big)\\
&&+\sum_{i}\tilde{f}_{\mathbf{x}}^{(i)}((y_{i})_{p})+\langle\nabla\tilde{f}_{\mathbf{x}}^{(i)}\left((y_{i})_{p}\right),(y_{i})_{q}-(y_{i})_{p} \rangle -\tilde{f}_{\mathbf{x}}^{(i)}((y_{i})_{q})+2k. 
\end{eqnarray*}
\end{theorem}

By Theorem~\ref{thm:rate}, the duality gap converges at the rate $\mathcal{O}\left(\sqrt{\operatorname{dist}\left(\mathbf{z}_{t}, \mathcal{Z}^{*}\right)}\right)$ in Lemma~\ref{lemma:3} with some approximation errors characterized by the Bregman divergence in Proposition~\ref{thm:distance} and capacity~$k$ much smaller than node $n$ in Proposition~\ref{thm:approx}.

\section{Discussions and Future Work}
One does not necessarily have to use linear objectives. For example, the objective of sum of the node players' costs is not a linear one.
The size of the selected subset for each player is also a model choice. Currently, both players have the same ``strength" in terms of $\|\mathbf{x}\|_1=\|\mathbf{y}\|_1=k$. It would be interesting to consider if the selected subsets have different sizes.

As future directions, we can generalize competitive opinion optimization to multi-player non-zero-sum games with different (linear) objectives in terms of expressed opinions for different players each optimizing her own objective.
Playing certain no-regret algorithms, the average strategy of each player then might converge to certain more permissive equilibrium (Nash equilibrium, correlated equilibrium, etc.).
It does not really make sense in a zero-sum game to ask about the price of anarchy.
Nevertheless, the price-of-anarchy type of questions becomes interesting and meaningful in a non-zero-sum game setting again.

\bibliographystyle{plain}
\bibliography{competitive}

\appendix
\section{Computing Matrix $Q_{UB}$} \label{app:a}
We restate the computation from Section 3.3 of \cite{gionis}.
The transition matrix $P$ is constructed by normalizing each row vector of the weight matrix $W$.
Given the set of absorbing nodes $B$ and the set of transient nodes $U$,
then $P$ can be partitioned into submatrices $P_{UB}$, $P_{UU}$, identity matrix $I$, and all-zero matrix $\mathbf{0}$,
where $P_{UB}$ is the $|U|\times|B|$ submatrix with the transition probabilities from transient nodes to absorbing nodes
and $P_{UU}$ is the $|U|\times|U|$ submatrix with the transition probabilities between transient nodes.

The probability of transition from $i$ to $j$ in exactly $l$ steps is denoted as the $(i,j)$ entry of the matrix $P_{UU}^l$.
We can construct the $|U|\times|U|$ \emph{fundamental matrix} $F$ of the absorbing random walk where the $(i,j)$ entry is the probability that such random walk starting from $i$ ends up at $j$ without being absorbed.
\[F=\sum_{t=0}^\infty(P_{UU})^l=(I-P_{UU})^{-1}.\]
Finally,  we have that
\[Q_{UB}=FP_{UB},\]
where each entry ${Q_{UB}}_{i,j}$ of such $|U|\times|B|$ matrix is the probability that a random walk starting at transient node $i$ gets absorbed at absorbing node $j$.

\section{Proof of Theorem~\ref{thm:minmax}} \label{app:b}
For every $t$, we have $\mathbf{x}^{(t)} \leftarrow \mathcal{L}_\mathcal{X}(g^{(1)},\ldots,g^{(t-1)})$ and $\mathbf{y}^{(t)} \leftarrow \mathcal{L}_\mathcal{Y}(h^{(1)},\ldots,h^{(t-1)})$ for $g^{(t)}= M \mathbf{y}_t$ and $h^{(t)}= -\mathbf{x}_t M$.
By applying the definition of regret twice, we have
\begin{eqnarray}
\frac{1}{T}\sum_{t=1}^T{\mathbf{x}^{(t)}}^T M \mathbf{y}^{(t)}&=&\min_{\mathbf{x}\in \mathcal{X}}\mathbf{x}^T M(\frac{1}{T}\sum_{t=1}^T \mathbf{y}^{(t)})+\frac{\operatorname{Regret}(\mathcal{L}_\mathcal{X})}{T} \nonumber\leq\max_{\mathbf{y}\in Y}\min_{\mathbf{x}\in X}\mathbf{x}^T M \mathbf{y} +\frac{o(T)}{T},\\
\frac{1}{T}\sum_{t=1}^T{\mathbf{x}^{(t)}}^T M \mathbf{y}^{(t)}&=&\max_{\mathbf{y}\in \mathcal{Y}}(\frac{1}{T}\sum_{t=1}^T \mathbf{x}^{(t)})M \mathbf{y}-\frac{\operatorname{Regret}(\mathcal{L}_\mathcal{Y})}{T} \nonumber\geq\min_{\mathbf{x}\in \mathcal{X}}\max_{\mathbf{y}\in \mathcal{Y}}\mathbf{x}^T M \mathbf{y} -\frac{o(T)}{T}.
\end{eqnarray}
One can obtain $\min_{\mathbf{x}\in \mathcal{X}}\max_{\mathbf{y}\in \mathcal{Y}}\mathbf{x}^T M \mathbf{y}\leq\max_{\mathbf{y}\in Y}\min_{\mathbf{x}\in \mathcal{X}}\mathbf{x}^T M \mathbf{y}+\frac{o(T)}{T}$ by combining the inequalities above and setting $T\rightarrow\infty$,
and $\min_{\mathbf{x}\in \mathcal{X}}\max_{\mathbf{y}\in \mathcal{Y}}\mathbf{x}^T M \mathbf{y}\geq\max_{\mathbf{y}\in \mathcal{Y}}\min_{\mathbf{x}\in \mathcal{X}}\mathbf{x}^T M \mathbf{y}$ by weak duality.

\section{Estimating the Adversary's Strategy} \label{app:g}

We now show that the adversarial player's strategy $\mathbf{y}^{\prime}$ satisfies with high probability
\begin{eqnarray*}
f(\mathbf{E}_{\mathbf{x}^{(t)}\sim X^{(t)}}[\mathbf{x}^{(t)}],\mathbf{y}^{(t)})\geq\arg\max_{\mathbf{y}\in \mathcal{Y}}f(\mathbf{E}_{\mathbf{x}^{(t)}\sim X^{(t)}}[\mathbf{x}^{(t)}],\mathbf{y})-\epsilon,
\end{eqnarray*}
where $\epsilon>0$ is an error from estimation,
can be efficiently computed (and thereby finding the loss function $f^{(t)}$) by \emph{estimating the probability that the min player selects each node~$i$}.

Recall that the cost value can be computed as
\begin{eqnarray*}
f(\mathbf{E}_{\mathbf{x}^{(t)}\sim X^{(t)}}[\mathbf{x}^{(t)}],\mathbf{y})&=&C(Q_{UB}({\mathbf{E}_{\mathbf{x}^{(t)}\sim X^{(t)}}[\mathbf{s}^{(t)}}']+\mathbf{y}))=\ell^\mathbf{T}(\mathbf{E}_{\mathbf{x}^{(t)}\sim X^{(t)}}[{\mathbf{s}^{(t)}}']+\mathbf{y}),
\end{eqnarray*}
where $\mathbf{E}_{\mathbf{x}^{(t)}\sim X^{(t)}}[{\mathbf{s}^{(t)}}']=\mathbf{s}+\mathbf{E}_{\mathbf{x}^{(t)}\sim X^{(t)}}[\mathbf{x}^{(t)}]$ is the expected modified internal opinion by the min player.
The randomized pure strategy produced by the follow-the-perturbed-leader algorithm at that time step provides a randomized way to modify entries of the vector $\mathbf{s}$ to the value $-1$.

Let $p^{(t)}_i$ denote the probability that it chooses to modify node $i$ at time step $t$.
Note that $\mathbf{E}_{\mathbf{x}^{(t)}\sim \mathbf{x}^{(t)}}[\mathbf{x}^{(t)}]=(p^{(t)}_i\cdot(-s_i-1))_i$.
For each $t$, we can draw $r$ samples, each of which is a $k$-subset of nodes, by using randomness of $R_t$ (Step~4 in Algorithm~\ref{alg:1}). 
We let $\hat{p}^{(t)}_i$ denote the ratio defined as the number of samples in which it chooses to modify node $i$ to the total number of samples $r$.
By applying the Hoeffding's inequality, we have
\[\mathbf{Pr}[|\hat{p}^{(t)}_i - p^{(t)}_i|>\epsilon]\leq 2\exp(-2\epsilon^2 r).\]
That is, by choosing $r=T_t$ where $T_t$ is the sample size for times step~$t$, we can use the estimated probability $\hat{p}^{(t)}_i$ that is within an estimation error $\epsilon=\sqrt{\frac{\ln T_t}{T_t}}$ from the actual one with at least probability of $1-\frac{2}{T_t^2}$.

Then the expected cost of the min player (before the adversarial player's intervention) is
    \[\sum_i \ell_i (\hat{p}^{(t)}_i(-1) + (1-\hat{p}^{(t)}_i)s_i).\]
The adversarial player would like to increase the min player's cost as much as possible.
For the max player, by compromising node $i$, the expected cost can be increased by
    \[\Delta_i=\ell_i \cdot 1 - \ell_i (\hat{p}^{(t)}_i(-1) + (1-\hat{p}^{(t)}_i)s_i).\]
Thus, the adversarial (max) player simply chooses the $k$ nodes with the $k$ largest $\Delta_i$'s.
Note that using $\hat{p}^{(t)}_i$ for each node $i$ incurs an estimation error that jointly guarantees computing the adversary's $\mathbf{y}^{(t)}$ efficiently such that with at least probability of $1-\frac{2}{T_t^2}$
\begin{eqnarray} 
f(\mathbf{E}_{\mathbf{x}^{(t)}\sim \mathbf{x}^{(t)}}[\mathbf{x}^{(t)}],\mathbf{y}^{(t)})\nonumber\geq\arg\max_{\mathbf{y}\in \mathcal{Y}}f(\mathbf{E}_{\mathbf{x}^{(t)}\sim \mathbf{x}^{(t)}}[\mathbf{x}^{(t)}],\mathbf{y})-\sqrt{\frac{\ln T_t}{T_t}}.
\end{eqnarray}

\section{Proof of Lemma~\ref{eq:no-regret}} \label{app:h}
We apply Theorem 1.1(a) of \cite{kalai:vempala}, which we restate in the following, 
in our context with random vector $R_t\in[0,\sqrt{T}]^n$ chosen uniformly at random in each dimension.
\begin{theorem}[Theorem 1.1(a) of \cite{kalai:vempala}] \label{thm:fpl}
Let $s_1,\ldots,s_T\in S$ be a state sequence (i.e., loss sequence in our terminology). Running the additive version of the follow-the-perturbed-leader algorithm (FPL) with parameter $\varepsilon\leq 1$ (i.e., learning rate) gives,
$$\mathbf{E}[\mbox{cost of FPL(}\varepsilon)]\leq\mbox{min-cost}_T+\varepsilon RAT+\frac{D}{\varepsilon},$$
where $D$ is the diameter of the decision space, $R$ is an upper bound on a loss value, and $A$ is an upper bound on the $L_1$-norm of a loss vector. It makes sense to state the bound in terms of the minimizing value of $\varepsilon$, giving
$$\mathbf{E}[\mbox{cost of FPL(}\frac{D}{\sqrt{RAT}})]\leq\mbox{min-cost}_T+2\sqrt{DRAT}.$$
\end{theorem}
In our context, we can choose $D=n, R=n, A=n$ so $\varepsilon=\frac{1}{\sqrt{T}}$, and the decision space is $\mathcal{X}$. By Theorem~\ref{thm:fpl}, the first inequality holds for some constant $0<\delta\leq 1$
\begin{eqnarray*}
&&\frac{1}{T}\sum_{t=1}^T \mathbf{E}_{\mathbf{x}^{(t)}\sim X^{(t)}}[f^{(t)}(\mathbf{x}^{(t)})]=\frac{1}{T}\sum_{t=1}^T E_{\mathbf{x}^{(t)}\sim X^{(t)}}[f(\mathbf{x}^{(t)},\mathbf{y}^{(t)})]\\
&\leq&(1+\delta)(\frac{1}{T}\min_{\mathbf{x}\in \mathcal{X}}\sum_{t=1}^T f^{(t)}(\mathbf{x})+\frac{2n^{3/2}}{\sqrt{T}})\\
&=&(1+\delta)\frac{1}{T}\min_{\mathbf{x}\in \mathcal{X}}\sum_{t=1}^T f(\mathbf{x},\mathbf{y}^{(t)})+\frac{2(1+\delta)n^{3/2}}{\sqrt{T}}.
\end{eqnarray*}

\section{Proof of Theorem~\ref{thm:strategy}} \label{app:c}
For the min player, we have that
\begin{eqnarray*}
&&\max_{\mathbf{y}\in \mathcal{Y}}f(\mathbf{E}_{T_{\min}\in\{1,\ldots,T\},\mathbf{x}^{(T_{\min})}\sim X^{(T_{\min})}}[\mathbf{x}^{(T_{\min})}],\mathbf{y})\\
&=&\max_{\mathbf{y}\in \mathcal{Y}}\frac{1}{T}\sum_{t=1}^T f(\mathbf{E}_{\mathbf{x}^{(t)}\sim \mathbf{x}^{(t)}}[\mathbf{x}^{(t)}],\mathbf{y})\leq\frac{1}{T}\sum_{t=1}^T\max_{\mathbf{y}\in \mathcal{Y}}f(\mathbf{E}_{\mathbf{x}^{(t)}\sim \mathbf{x}^{(t)}}[\mathbf{x}^{(t)}],\mathbf{y})
\end{eqnarray*}
by the linearity of expectation and for any $\mathbf{y}'$
\[\max_{\mathbf{y}\in \mathcal{Y}}f(\mathbf{E}_{\mathbf{x}^{(t)}\sim X^{t}}[\mathbf{x}^{(t)}],\mathbf{y})\geq f(\mathbf{E}_{\mathbf{x}^{(t)}\sim X^{t}}[\mathbf{x}^{(t)}],\mathbf{y}').\]

For each $t$, applying Proposition~\ref{pro:adversary} that accounts for estimation of the adversary's strategy with an estimation error $\epsilon=\sqrt{\frac{\ln T}{T}}$, we obtain with probability of at least $1-2/T$ (by a union bound)
\begin{eqnarray*}
\frac{1}{T}\sum_{t=1}^T\max_{\mathbf{y}\in \mathcal{Y}}f(\mathbf{E}_{\mathbf{x}^{(t)}\sim \mathbf{x}^{(t)}}[\mathbf{x}^{(t)}],\mathbf{y})\leq\frac{1}{T}\sum_{t=1}^T f(\mathbf{E}_{\mathbf{x}^{(t)}\sim \mathbf{x}^{(t)}}[\mathbf{x}^{(t)}],\mathbf{y}^{(t)})+\sqrt{\frac{\ln T}{T}}.
\end{eqnarray*}
Due to the fact that $f$ is affine in $\mathbf{x}^{(t)}$, the right-hand side of the inequality is equivalent to
\begin{eqnarray*}
\frac{1}{T}\sum_{t=1}^T\mathbf{E}_{\mathbf{x}^{(t)}\sim \mathbf{x}^{(t)}}[f(\mathbf{x}^{(t)},\mathbf{y}^{(t)})]+\sqrt{\frac{\ln T}{T}}.
\end{eqnarray*}

By the $O(\frac{1}{\sqrt{T}})$-average regret property from Lemma \ref{eq:no-regret},
we finally have with probability of at least $1-\frac{2}{T}$ for some constant $0<\delta\leq 1$
\begin{eqnarray*}
&&\frac{1}{T}\sum_{t=1}^T\mathbf{E}_{\mathbf{x}^{(t)}\sim \mathbf{x}^{(t)}}[f(\mathbf{x}^{(t)},\mathbf{y}^{(t)})]+\sqrt{\frac{\ln T}{T}}\\
&\leq&(1+\delta)(\frac{1}{T}\min_{\mathbf{x}\in \mathcal{X}}\sum_{t=1}^T g(\mathbf{x},\mathbf{y}^{(t)})+\frac{2n^{3/2}}{\sqrt{T}})+\sqrt{\frac{\ln T}{T}}
\leq(1+\delta)\min_{\mathbf{x}\in \mathcal{X}}\max_{\mathbf{y}\in \mathcal{Y}}f(\mathbf{x},\mathbf{y})+\frac{\sqrt{\ln T}+2(1+\delta)n^{3/2}}{\sqrt{T}}.
\end{eqnarray*}
In summary, for the min player, with probability of at least $1-\frac{c''}{n}-\frac{2}{T}$
\[\max_{\mathbf{y}\in \mathcal{Y}}f(\mathbf{E}_{T_{\min}\in\{1,\ldots,T\},\mathbf{x}^{(T_{\min})}\sim X^{(T_{\min})}}[\mathbf{x}^{(T_{\min})}],\mathbf{y})\leq(1+\delta)\min_{\mathbf{x}\in \mathcal{X}}\max_{\mathbf{y}\in \mathcal{Y}}f(\mathbf{x},\mathbf{y})+\frac{\sqrt{\ln T}+2(1+\delta)n^{3/2}}{\sqrt{T}}.\]

\section{Proof of Proposition~\ref{thm:approx}} \label{app:d}
The difference between $f_{\mathbf{y}}$ and $\bar{f}_{\mathbf{y}}$ is maximized at $0$ so
$\sum_{i}(\ell_{i} \cdot\left(s_{i}+0+y_{i}\right)-\ell_{i}) \leq\sum_{i}\left(\ell_{i} \cdot\left(s_{i}+y_{i}\right)-\ell_{i}\right) \leq\sum_{i}\left(\ell_{i} \cdot\left(s_{i}+y_{i}-1\right)\right)\leq\sum_{i}\left(s_{i}+y_{i}-1\right) \leq \sum_{i} s_{i}+\sum_{i} y_{i}-n\leq$ $n+k-n=k<<n$, where $0 \leq \ell_{i} \leq 1$ by the property that $Q_{U B}{ }_{i, j}$ is the probability from transient $i$ to absorbing $j$. Similarly, under the condition that $x$ is fixed, the corresponding conclusion can also be obtained $ f_{\mathbf{x}}(y)-\bar{f}_{x}(y) \leq k$ for all $y$. 

\section{Proof of Proposition~\ref{thm:strong_conv}} \label{app:e}
By the strong convexity for each $x_{i}$,
$$
\sum_{i}(\tilde{f}(\mathbf{x}, \mathbf{y})-\tilde{f}(x_{i}^{*}, x_{-i}, \mathbf{y})) \leq \sum_{i}(\nabla_{x_{i}} \tilde{f}(\mathbf{x},\mathbf{y})^{\top}(x_{i}-x_{i}^{*})-\gamma_{i}\|x_{i}-x_{i}^{*}\|^{2} / 2).
$$
We have
\begin{eqnarray*}
&&\sum_{i}(\tilde{f}(\mathbf{x}, \mathbf{y})-\tilde{f}(x_{i}^{*}, x_{-i}, \mathbf{y}))\\ 
&=&\sum_{i}(\sum_{j \neq-i} \ell_{j}(s_{j}+x_{j}+y_{j})+\ell_{i}(s_{i}+x_{i}+y_{i})-\sum_{j \neq-i} \ell_{j}(s_{j}+x_{j}+y_{j})-\ell_{i}(s_{i}+x_{i}^{*}+y_{i})) \\
&=&\sum_{i}(\ell_{i}(s_{i}+x_{i}+y_{i})-\ell_{i}(s_{i}+x_{i}^{*}+y_{i}))\\
&=&\tilde{f}(\mathbf{x}, \mathbf{y})-\tilde{f}(\mathbf{x}^{*}, \mathbf{y})
\end{eqnarray*}
which is the LHS of the inequality. For the RHS,
$\sum_{i}(\nabla_{x_{i}} \tilde{f}(x,y)^{\top}(x_{i}-x_{i}^{*})=\nabla_x\tilde{f}(\mathbf{x}, \mathbf{y})^{\top}(\mathbf{x}-\mathbf{x}^{*})$ and
$$
\sum_{i} \gamma_{i}\left\|x_{i}-x_{i}^{*}\right\|^{2} \geq \min _{i} \gamma_{i} \sum_{i}\left\|x_{i}-x_{i}^{*}\right\|^{2}\eqno{(5)}
$$
Thus, $\gamma$ can be set to $\min _{i} \gamma_{i}$, and the claimed inequality holds. 


\section{Average-Iterate Convergence} \label{app:f}



A desirable property of online convex learning/optimization algorithms is to have a sub-linear scaling regret rate in~$T$. 
This property implies that the regret per round goes to zero as $T$ goes to infinity. For example, ${\operatorname{Regret}\left(\mathcal{L}\right)}/ T=\frac{\sqrt{T}}{T} \rightarrow 0$ when $T \rightarrow \infty$. Using the property of no-regret algorithms and the definition of regret, we would like to prove the average duality gap convergence for our game setting.

In our case of two-player zero-sum games, the regret of the min player ${\operatorname{Regret}\left(\mathcal{L}_{\mathcal{X}}\right)}$ is defined according Definition 1 as: 
$$
{\operatorname{Regret}\left(\mathcal{L}_{\mathcal{X}}\right)} = \sum_{t=1}^{T} \tilde{f}\left(\mathbf{x}^{(t)}, \mathbf{y}^{(t)}\right)-\min _{\mathbf{x}^{\prime} \in X} \sum_{t=1}^{T} \tilde{f}\left(\mathbf{x}^{\prime}, \mathbf{y}^{(t)}\right),
$$
and the regret ${\operatorname{Regret}\left(\mathcal{L}_{\mathcal{Y}}\right)}$ of max player can also be defined as 
$$
{\operatorname{Regret}\left(\mathcal{L}_{\mathcal{Y}}\right)} = \max_{\mathbf{y}^{\prime} \in Y} \sum_{t=1}^{T} \tilde{f}\left(\mathbf{x}^{(t)}, \mathbf{y}^{\prime}\right)-\sum_{t=1}^{T} \tilde{f}\left(\mathbf{x}^{(t)}, \mathbf{y}^{(t)}\right).
$$
Our game is played repeatedly, for every $t$, we have $\mathbf{x}^{(t)} \leftarrow \mathcal{L}_{\mathcal{X}}\left(g^{(1)}, \ldots, g^{(t-1)}\right)$ and $\mathbf{y}^{(t)} \leftarrow \mathcal{L}_{\mathcal{Y}}\left(h^{(1)}, \ldots, h^{(t-1)}\right)$, for $g^{(t)}=\nabla_{\mathbf{x}} f(\mathbf{\mathbf{x}^{(t)}} ,\mathbf{y^{(t)}})$ and $h^{(t)} =-\nabla_{\mathbf{y}} f(\mathbf{\mathbf{x}^{(t)}},\mathbf{y^{(t)}})$, with the $o(T)$-regret algorithms $\mathcal{L}_{\mathcal{X}}$ and $\mathcal{L}_{\mathcal{Y}}$ that two players play, respectively, in a zero-sum game.

Applying the definition of regret for each $\mathcal{L}_{\mathcal{X}}$ and $\mathcal{L}_{\mathcal{Y}}$, and take the average of regret by $T$ . We obtain
$$
\begin{aligned}
&\frac{1}{T} \sum_{t=1}^{T} \tilde{f}(\mathbf{\mathbf{x}^{(t)}} ,\mathbf{y^{(t)}})=\max _{\mathbf{y}^{\prime} \in \mathcal{Y}} \frac{1}{T} \sum_{t=1}^{T}\tilde{f}(\mathbf{\mathbf{x}^{(t)}} ,\mathbf{y}^{\prime})-\frac{\operatorname{Regret}\left(\mathcal{L}_{\mathcal{Y}}\right)}{T}=\max _{\mathbf{y}^{\prime} \in \mathcal{Y}} \frac{1}{T} \sum_{t=1}^{T}\tilde{f}(\mathbf{\mathbf{x}^{(t)}} ,\mathbf{y}^{\prime})-\frac{o(T)}{T}  , \\
&\frac{1}{T} \sum_{t=1}^{T} \tilde{f}(\mathbf{x}^{(t)} ,\mathbf{y^{(t)}})=\min _{\mathbf{x}^{\prime} \in \mathcal{X}} \frac{1}{T} \sum_{t=1}^{T}\tilde{f}(\mathbf{x}^{\prime} ,\mathbf{y^{(t)}})+\frac{\operatorname{Regret}\left(\mathcal{L}_{\mathcal{X}}\right)}{T}=\min _{\mathbf{x}^{\prime} \in \mathcal{X}} \frac{1}{T} \sum_{t=1}^{T}\tilde{f}(\mathbf{x}^{\prime} ,\mathbf{y^{(t)}})+\frac{o(T)}{T}.
\end{aligned}
$$

Combining the two above inequalities and adding up the approximation errors in Propositions 1 and 2, we can guarantee an ``average duality gap convergence”:

\begin{eqnarray*}
&&\frac{1}{T} \sum_{t=1}^{T} \max_{\mathbf{x}^{\prime}, \mathbf{y}^{\prime}}\tilde{f}\left(\mathbf{x}^{(t)}, \mathbf{y}^{\prime}\right)-\tilde{f}\left(\mathbf{x}^{\prime}, \mathbf{y}^{(t)}\right)+2k\\
&&+\sum_{i}\big((\tilde{f}_{\mathbf{x}}^{(i)}((y_{i})_{p})+\langle\nabla\tilde{f}_{\mathbf{x}}^{(i)}\left((y_{i})_{p}\right),(y_{i})_{q}-(y_{i})_{p} \rangle) -\tilde{f}_{\mathbf{x}}^{(i)}((y_{i})_{q})\big)\\
&&+\sum_{i}\big(\tilde{f}_{\mathbf{y}}^{(i)}\left((x_{i})_{q}\right)-\big(\tilde{f}_{\mathbf{y}}^{(i)}\left((x_{i})_{p}\right)+\langle\nabla\tilde{f}_{\mathbf{y}}^{(i)}\left((x_{i})_{p}\right), (x_{i})_{q}-(x_{i})_{p}\rangle\big)\big) \\
&=&\mathcal{O}\left(\frac{1}{\sqrt{T}}\right)+\frac{\sum_{t=1}^{T}\left(\left(D_{\tilde{f}}\left(y_{p}^{\prime}, y_{q}^{\prime}\right)+D_{\tilde{f}}\left(x_{p}^{\prime}, x_{q}^{\prime}\right)\right)\right)}{T}.
\end{eqnarray*}

We use notation $D_{\tilde{f}}\left(x_{p}^{\prime}, x_{q}^{\prime}\right)$ represents the maximum overall divergence of each dimension (  $\sum_{i}\big(\tilde{f}_{\mathbf{y}}^{(i)}\left((x_{i})_{q}\right)-\big(\tilde{f}_{\mathbf{y}}^{(i)}\left((x_{i})_{p}\right)+\langle\nabla\tilde{f}_{\mathbf{y}}^{(i)}\left((x_{i})_{p}\right), (x_{i})_{q}-(x_{i})_{p}\rangle\big)\big)$) when selected $\mathbf{x}^{\prime}$, and the notation $D_{\tilde{f}}\left(y_{p}^{\prime}, y_{q}^{\prime}\right)$ similarly represents the maximum overall divergence of each dimension ($\sum_{i}\big((\tilde{f}_{\mathbf{x}}^{(i)}((y_{i})_{p})+\langle\nabla\tilde{f}_{\mathbf{x}}^{(i)}\left((y_{i})_{p}\right),(y_{i})_{q}-(y_{i})_{p} \rangle) -\tilde{f}_{\mathbf{x}}^{(i)}((y_{i})_{q})\big)$) when selected $\mathbf{y}^{\prime}$, and $k$ is the limit of capacity for our game.
Here, we can conclude an average duality convergence gap of OGDA at a rate of $\mathcal{O}\left(\frac{1}{\sqrt{T}}\right)$ with some approximation errors after $T$ iterations.

\end{document}